\documentclass[11pt]{article}
\usepackage{amssymb}
\usepackage{amsmath}
\usepackage{amsthm}
\usepackage{color}
\usepackage[colorlinks,linkcolor=blue,citecolor=green]{hyperref}

\newtheorem{theorem}{Theorem}
\newtheorem{corollary}[theorem]{Corollary}

\newtheorem{lemma}[theorem]{Lemma}

\newtheorem{remark}[theorem]{Remark}
\newcommand*{\pd}
[2]{\mathchoice{\frac{\partial#1}{\partial#2}}
  {\partial#1/\partial#2}{\partial#1/\partial#2}
  {\partial#1/\partial#2}}

\newcommand*{\fd}
[2]{\mathchoice{\frac{\delta#1}{\delta#2}}
  {\delta #1/\delta#2}{\delta#1/\delta#2}{\delta#1/\delta#2}}

\title{\textbf{Remarks on\\ the Lagrangian representation\\
of bi-Hamiltonian equations}}
\author{
M.V. Pavlov$^{1,2,3}$ \\
$^{1}$Sector of Mathematical Physics,\\
Lebedev Physical Institute of Russian Academy of Sciences,\\
Leninskij Prospekt 53, 119991 Moscow, Russia\\
$^{2}$Department of Applied Mathematics,\\
National Research Nuclear University MEPHI,\\
Kashirskoe Shosse 31, 115409 Moscow, Russia\\
$^{3}$Novosibirsk State University,\\
2 Pirogova street, 630090 Novosibirsk, Russia\\
\url{maksmath@gmail.com}\\[3mm]
R.F. Vitolo$^{4}$ \\
$^{4}$Department of Mathematics and Physics \textquotedblleft E. De
Giorgi\textquotedblright ,\\
University of Salento, Lecce, Italy\\
\url{raffaele.vitolo@unisalento.it}\\
\url{http://poincare.unisalento.it/vitolo}}
\date{}

\begin{document}

\maketitle

\begin{abstract}
  The Lagrangian representation of multi-Hamiltonian PDEs has been introduced
  by Y. Nutku and one of us (MVP). In this paper we focus on systems which are
  (at least) bi-Hamiltonian by a pair $A_1$, $A_2$, where $A_1$ is a
  hydrodynamic-type Hamiltonian operator. We prove that finding the Lagrangian
  representation is equivalent to finding a generalized vector field $\tau$
  such that $A_2=L_\tau A_1$. We use this result in order to find the
  Lagrangian representation when $A_2$ is a homogeneous third-order Hamiltonian
  operator, although the method that we use can be applied to any other
  homogeneous Hamiltonian operator.  As an example we provide the Lagrangian
  representation of a WDVV hydrodynamic-type system in $3$ components.

\bigskip


\bigskip

\end{abstract}

\tableofcontents

\section{Introduction}

Integrable systems of nonlinear PDEs in $1+1$ dimensions (or, equivalently, in
$2$ independent variables) have a particularly rich structure. Evolutionary
systems of the form
\begin{equation}
  \label{eq:22}
  u^i_t = f^i(u^k,u^k_x,u^k_{xx},\ldots)
\end{equation}
can be given one or more Hamiltonian formulation
\begin{equation}
  \label{eq:10}
  u^i_{t} = f^i(u^k,u^k_x,u^k_{xx},\ldots) =
  A_{\alpha}^{ij}\fd{\mathbf{H}_{\alpha}}{u^j},\quad\text{(no sum in
    $\alpha$)},\quad i=1,\ldots,n
\end{equation}
through one or more Hamiltonian operators $A_\alpha$. Here, $\alpha$ is an
index running from $1$ to $M$, $\mathbf{H}_\alpha$ are Hamiltonian densities and
$\fd{}{u^j} = (-1)^{|\sigma|}\partial_\sigma\,\pd{}{u^i_\sigma}$ are
variational derivatives ($\sigma$ is a multi-index). The requirement of
$A_\alpha$ being Hamiltonian amounts at $A_\alpha$ being a formally
skew-adjoint differential operator in total derivatives $\partial_\sigma$ and
$[A_\alpha,A_\alpha] = 0$. The last operator is the Schouten bracket,
see~\cite{Dorf} for a definition. Hamiltonian operators map characteristic
vectors of conservation laws into generating functions of (generalized)
symmetries \cite{Dorf,olver}.

By a classical result by F. Magri (see \cite{Dorf} and references therein), the
integrability of the system~\eqref{eq:10} is ensured by the existence of at
least $2$ Hamiltonian operators that are compatible: $[A_1,A_2]=0$. This yields
a sequence of commuting conserved quantities and symmetries and, eventually,
leads to integration through the inverse scattering method.

An interesting feature of bi-Hamiltonian systems ($M=2$) or multi-Hamiltonian
systems ($M>2$) was first pointed out in \cite{Yavuz}. Suppose that the
Hamiltonian operator $A_\beta$ is invertible.  Then the operators
$R^\alpha_\beta=A_{\alpha}A_\beta^{-1}$ applied to eq.~\eqref{eq:10} yield,
after potential substitution $u^i=\varphi^i_x$, systems of PDEs which admit
local Lagrangians in the new variables $\varphi^i$. If all the operators
$A_\alpha$ are invertible there will be $2M-1$ new systems. This is the
\emph{Lagrangian representation} of the system~\eqref{eq:10}. It is clear that
having the possibility to represent a system of PDEs and its hierarchy of
symmetries by Lagrangian PDEs is mathematically interesting. This possibility
was considered by several authors (see, \emph{e.g.},\cite{OM98}), and was
systematically linked to the multi-Hamiltonian property
in~\cite{Yavuz}. Several examples have been considered so far: KdV, polytropic
gas dynamics, Kaup--Boussinesq, Kaup--Broer, Boussinesq, NLS~\cite{Yavuz}.

In this paper we will consider evolutionary systems of PDEs which are at least
bi-Hamiltonian by a pair of Hamiltonian operators $A_1$, $A_2$:
\begin{equation}\label{eq:381}
u_{t^{k}}^{i}=A_{2}^{is}\frac{\delta \mathbf{H}_{k}}{\delta u^{s}}%
=A_1^{is}\frac{\delta \mathbf{H}_{k+1}}{\delta u^{s}},
\end{equation}
where for $k=0$ we have the initial system of the hierarchy of commuting flows.
We make the assumption that the first Hamiltonian operator
$A_1$ is a first-order Dubrovin--Novikov homogeneous Hamiltonian operator
\begin{equation}
A_{1}^{ij}=g^{ij}\partial _{x}+b_{k}^{ij}u_{x}^{k},  \label{eq:192}
\end{equation}
where $g^{ij}$, $b_{k}^{ij}$ are functions of the field variables $(u^h)$ and
homogeneity is meant with respect to the grading that assigns the degree $1$ to
$\partial_x$.  Such operators are quite common among the bi-Hamiltonian
systems. Moreover, they are easily invertible: in the non-degenerate case
$\det(g^{ij})\neq 0$ there is always a transformation
$\bar{u}^i=\bar{u}^i(u^j)$ such that $A_1=K^{ij}\partial_x$, where $K^{ij}$ is
a constant symmetric matrix, hence $A_1^{-1}=(K^{ij})^{-1}\partial_x^{-1}$.

Bi-Hamiltonian systems of the above type are known to have the following
properties (besides admitting a Lagrangian representation).
\begin{enumerate}
\item They admit a nonlocal \emph{symplectic operator}
  $B=A_{1}^{-1}A_{2}A_{1}^{-1}$ \cite[Proposition 7.9]{Dorf}. This operator becomes
  local after the above potential substitution, and the above
  system~\eqref{eq:381} can be rewritten as
  \begin{equation}
    \label{eq:19}
    B_{im}\varphi _{t^{k}}^{m}=
      \frac{\delta \mathbf{H}_{k+2}}{\delta \varphi ^{i}},
  \end{equation}
  The above equation is evidently Lagrangian; its \emph{Lagrangian
    representation} is a Lagrangian for the above equation of the form
  \begin{equation}
    \label{eq:20}
    L_{n}\varphi _{t^{k}}^{n}-h_{k+2},
  \end{equation}
  where $h_{k+2}$ is the function that defines the Hamiltonian density
  $\mathbf{H}_{k+2}$ (which is a quantity that can be computed) and $(L_n)$ is
  a vector function of $\varphi^i$, $\varphi^i_x$, $\varphi^i_{xx}$, \dots,
  called the \emph{characteristic function} of the Lagrangian representation,
  to be determined.
\item The second Hamiltonian operator is the Lie derivative of the first one:
  $A_{2}=L_{\tau }A_{1}$, where $\tau $ is a generalized vector field (see
  \cite[Proposition 7.9]{Dorf} or, more generally
  \cite{DMS05,getz,DZ0}). According with a modern terminology \cite{DZ0,Lor},
  $A_2$ is a \emph{trivial infinitesimal deformation} of $A_1$.
\end{enumerate}

Clearly, the difficult part of finding the Lagrangian representation of the
equation~\eqref{eq:41} is finding $(L_n)$, as $h_{k+2}$ can be found from the
initial equation by bi-Hamiltonian recursion with usual methods. We stress that
the Lagrangian representation is \emph{local}, as there are topological
obstructions in finding global Lagrangians, see the discussion below.

In this paper we will prove the following:

\textbf{Theorem.}\ \emph{Finding the Lagrangian representation~\eqref{eq:42} of
  the integrable hierarchy~\eqref{eq:381} is equivalent to finding the
  generalized vector field $\tau $ such that $A_2 = L_\tau A_1$.}

More precisely, we will prove that the vector function $\psi= - (L_n)$ is
nothing but a potential of the symplectic form $B$ with respect to the
differential $e_1$ of the variational sequence \cite{Many,Dorf,olver,hand}.
Variational sequences have been introduced as the analogue of the de Rham
sequence for the calculus of variations. Let us summarize the features of this
topic that are relevant to this paper, see \cite{Many,Dorf,olver,hand} for much
deeper insight. Variational sequences can be schematically represented by the
sequence of differential operators
\begin{equation}
  \label{eq:23}
  \dots \overset{D}{\longrightarrow} \text{Lagrangians}
  \overset{\mathcal{E}}{\longrightarrow} \text{Variational 1-forms}
  \overset{\mathcal{H}}{\longrightarrow} \text{Variational 2-forms}\dots
\end{equation}
where $D$ is the total divergence, $\mathcal{E}$ is the Euler-Lagrange
operator, $\mathcal{H}$ is the Helmholtz operator any two consequent operations
are identically zero. Variational $1$-forms are covector-valued densities
$\psi$ which define systems of PDEs of the form $\psi=0$. If $H(\psi)=0$ then,
locally, the system $\psi=0$ is Lagrangian: locally there exists a Lagrangian
density $L$ such that $\mathcal{E}(L) = \psi$, and $L$ can be computed by the
Volterra homotopy operator \cite{Many,Dorf,olver,hand}.

The spaces of variational $k$-forms (variational $0$-forms are Lagrangians)
form sheaves, and the variational sequence turns out to be a locally exact
sequence of sheaves whose differentials are all denoted by $e_1$. Then, a
variational $2$-form $B$ is symplectic iff $e_1(B)=0$. If this is true, by
local exactness there exists a variational $1$-form $\psi$ (which is in general
only locally defined) such that $H(\psi) = B$. In the proof of the above
Theorem it will be shown that $\tau =A_{1}(\psi)$, therefore the Lagrangian
representation will have the same local (or global) character as $\psi$.  There
are well-known topological obstructions to the global exactness of the
variational sequence \cite{Many,Dorf,olver,hand}.

The above results hold in a differential-geometric setting, \emph{i.e.} for
$\mathcal{C}^\infty$ forms and vector fields, despite the fact that they have
initially been introduced having polynomial categories in mind \cite{Dorf}.

\medskip

With the above ideas in mind we found new nontrivial examples of Lagrangian
representation, besides those that were discussed so far, in the class of
bi-Hamiltonian systems where $A_2$ is a third-order homogeneous Hamiltonian
operator. In this class the homotopy operator of the variational sequence
allows us to choose a vector function $(L_n)$ which is: 1 - homogeneous, since
the homotopy operator preserves homogeneity, and 2 - dependent on third-order
derivatives, since the homotopy operator preserves the order of derivatives.

However, we had to overcome an obstacle. We are forced to compute in flat
coordinates of the first Hamiltonian operator $A_1$, as this makes $A_1$ easily
invertible. In these coordinates the coefficients of the homogeneous operator
$A_2$ yield non-removable singularities in the homotopy operator. So, we had to
develop a different method that, in principle, can be used with any homogeneous
operator $A_2$ of arbitrary order. We proved the following:

\textbf{Theorem.}\ {\itshape There exists a Lagrangian
  representation~\eqref{eq:20} where $L_n$ is a homogeneous polynomial of
  derivatives of degree $2$ if and only if the integrability condition
  $dT=0$ is fulfilled, where $T$ is a three-form defined in~\eqref{eq:21}.
  In this case we have
  \begin{equation}\label{eq:13}
    L_{n}=\left( \frac{1}{2}G_{nm}u_{x}^{m}+R_{nm}u_{x}^{m}\right) _{x} -
     \frac{1}{2}L_{nsm}u_{x}^{s}u_{x}^{m}
  \end{equation}
  where $(u^i)$ are flat coordinates of $g$, $G_{nm}$ is the leading
  coefficient of the symplectic operator $B$, $L_{nsm}u_{x}^{s}u_{x}^{m}$ are
  $n$ conservation law densities of the initial system of PDEs and $R_{nm}$ is
  a skew-symmetric tensor.  }

In our method we made a minimal order ansatz for the vector function $(L_n)$:
it can be a homogeneous polynomial of order not lower than $2$. This simplifies
computations at the cost of leading us to the integrability condition $dT=0$.

The above results have been successfully used to compute the Lagrangian
representation for the WDVV system of hydrodynamic type~\eqref{eq:5} for which
the bi-Hamiltonian nature is known \cite{FGMN}. The procedure can be applied
with equal ease to the bi-Hamiltonian systems \cite{KN1,KN2,pv}, thus providing
a wide range of examples. In all the above cases the integrability condition is
fulfilled and it is conjectured that the condition is a consequence of the fact
that $B$ is a symplectic operator. See the Conclusions for more details.

Symbolic computations were done in Reduce, a free Computer Algebra System,
using the package CDE, developed by the author of this paper
\cite{cde}. Interested readers are warmly invited to contact the author for
question on any aspect of the computations or to get software and/or
mathematical expressions of quantities that would be unpractical to include in
the paper.

\section{Lagrangian representation and triviality of $A_2$}
\label{sec:lagr-repr-bi}

In this section we will recall the construction of the Lagrangian
representation of a bi-Hamiltonian system in the case when $A_1$ is a
first-order Dubrovin--Novikov homogeneous Hamiltonian operator. Then we will
show the relationship between the (local) existence of the Lagrangian
representation and the (local) existence of a generalized vector field $\tau$
such that $A_2 = L_\tau A_1$.

Let us recall how the Lagrangian representation is defined in our framework. We
work in flat coordinates $(u^i)$ of $A_1$, so that $A_1=K^{ij}\partial_x$.
After the potential substitution $u^{i}=\varphi _{x}^{i}$ the integrable
hierarchy~\eqref{eq:381}
becomes
\begin{equation}\label{eq:39}
\varphi _{xt^{k}}^{i}=-A_{2}^{is}\partial _{x}^{-1}\frac{\delta \mathbf{H}%
_{k}}{\delta \varphi ^{s}}=-K^{is}\frac{\delta \mathbf{H}_{k+1}}{\delta
\varphi ^{s}}.
\end{equation}
Then we take two copies of this integrable hierarchy (here
$M_{ij}K^{js}=\delta^s_i$)
\begin{gather*}
 -M_{im}\varphi _{xt^{k}}^{m}=M_{im}A_{2}^{ms}\partial _{x}^{-1}\frac{\delta
\mathbf{H}_{k}}{\delta \varphi ^{s}}=\frac{\delta \mathbf{H}_{k+1}}{\delta
\varphi ^{i}},
\\
 -M_{im}\varphi
_{xt^{k+1}}^{m}=M_{im}A_{2}^{ms}\partial _{x}^{-1}\frac{\delta \mathbf{H}%
_{k+1}}{\delta \varphi ^{s}}=\frac{\delta \mathbf{H}_{k+2}}{\delta \varphi
^{i}},
\end{gather*}
which leads to two series of recursive relationships%
\begin{equation*}
\frac{\delta \mathbf{H}_{k+1}}{\delta \varphi ^{i}}=M_{ip}A_{2}^{pq}\partial
_{x}^{-1}\frac{\delta \mathbf{H}_{k}}{\delta \varphi ^{q}},\text{ \ }\varphi
_{xt^{k+1}}^{i}=A_{2}^{ip}M_{pq}\varphi _{t^{k}}^{q}.
\end{equation*}%
Thus%
\begin{equation*}
\frac{\delta \mathbf{H}_{k+2}}{\delta \varphi ^{i}}=-M_{im}\varphi
_{xt^{k+1}}^{m}=-M_{im}A_{2}^{mp}M_{pq}\varphi _{t^{k}}^{q}.
\end{equation*}
This means that the above equations are nothing but Euler--Lagrange equations
\begin{gather}\label{eq:41}
B_{im}\varphi _{t^{k}}^{m}= \frac{\delta \mathbf{H}_{k+2}}{\delta \varphi ^{i}},
\\
B_{ij} = - M_{ip}A_{2}^{pq}M_{qj}\label{eq:9}
\end{gather}
where $B=(B_{ij})$ is a local differential operator. It is not difficult to
prove that $B$ is a symplectic operator \cite{Dorf}.

We observe that the system \eqref{eq:39} is Lagrangian with respect to the
action density $\frac{1}{2}M_{ij}\varphi _{x}^{i}\varphi _{t^{k}}^{j}-h_{k+1}
$, where $H_{k}=\int h_{k}dx$. The equation~\eqref{eq:41} is Lagrangian too;
this fact is characteristic of bi-Hamiltonian systems and was systematically
investigated in \cite{Yavuz}. \emph{The Lagrangian representation of the
  bi-Hamiltonian system \eqref{eq:381}} is the Lagrangian of the
system~\eqref{eq:41}, which is of the type
\begin{equation}
L_{n}\varphi _{t^{k}}^{n}-h_{k+2}.  \label{eq:42}
\end{equation}

Now, we would like to prove one of the main results of this paper, \emph{ie}
the fact that the existence of a Lagrangian representation is equivalent to the
triviality of $A_2$ as an infinitesimal deformation of $A_1$, or $A_{2}=L_{\tau
}A_{1}$.

Let us recall the notion of Lie derivative for Hamiltonian operators. We denote
the derivative coordinates $u^i_{x\cdots x}$ by $u^i_\sigma$, where the index
$\sigma $ is the order of the $x$-derivative.  Given any vector function
$F=(F_i)$, $i=1,\dots,n$ we denote the linearization (or Fr\'echet derivative)
of $F$ and its formal adjoint, respectively acting on a generalized vector
field $\tau=\tau^i\pd{}{u^i}$ and on a covector $\psi=\psi_jdu^j$, where
$\tau^i=\tau^i(u^k_\sigma)$ and $\psi_j=\psi_j(u^k_\sigma)$, by
\begin{displaymath}
   \ell_F(\tau) = \pd{F^k}{u^i_\sigma}\partial_\sigma \tau^i, \quad
   \ell_F^*(\psi) = (-1)^{|\sigma|} \partial_\sigma
     \left(\pd{F^k}{u^i_\sigma} \psi_k\right).
\end{displaymath}
The above definition extends also to matrix differential operators: if
$A(X)=A^{i\sigma}_j\partial_\sigma X^j$, then we have the differential operator
in two arguments $\ell_{A,X}(Y)= (\pd{A^{i\sigma}_j}{u^k_\mu})\partial_\sigma
X^j\partial_\mu Y^k$. Then
\begin{displaymath}
L_\tau
A_1(\psi)= - [A_1,\tau](\psi)= \ell_{A_1,\psi}(\tau) - \ell_\tau(A_1(\psi))  -
 A_1(\ell_\tau^*(\psi)),
\end{displaymath}
where $\psi$ is an arbitrary vector-valued density
and $[A_1,\tau]$ is the Schouten bracket (see, for example, \cite{Dorf}).

It is known \cite[Theorem 7.9]{Dorf} that, if $A_1$ and $A_2$ is a pair of
compatible Hamiltonian operators where $A_1$ is invertible, then
$B=A_1^{-1}A_2A^{-1}_1$ is a symplectic operator. Moreover, if $\psi$ is a
potential of $B$ in the variational sequence, \emph{i.e.} $e_1(\psi) = B$, then
the vector field $\tau =A_{1}(\psi )$ yields $A_{2}=L_{\tau }A_{1}$. This fact
can be regarded as a direct consequence of a more general fact, namely, the
vanishing of the Lichnerowicz--Poisson cohomology (which is defined by the
differential $d_{1}=[A_{1},\cdot ]$) for first-order homogeneous operators. See
\cite{getz,DMS05}, or \cite[Lemma 2.4.20]{DZ0} for a more recent proof.  See
also \cite{Art} for a complete discussion about the problem of characterizing
which operators of the form $L_{\tau }A_{1}$ are Hamiltonian.

There is a deep link between the Lagrangian representation, the symplectic form
$B$ and the fact that $A_2$ can be expressed as the Lie derivative of $A_1$.

\begin{theorem}
  \label{th:symplag} The systems \eqref{eq:41} admit a Lagrangian $L_{n}\varphi
  _{t^{k}}^{n}-h_{k+2}$ where $\psi = - (L_{n})$ is a potential of the symplectic
  form in the variational sequence: $e_{1}(\psi )=B$. Moreover, we have
  \begin{equation}
    \label{eq:14}
     A_2 = L_\tau A_1,\quad\text{where}\quad \tau=A_1(\psi) =
     - K^{in}L_n\pd{}{u^i}.
  \end{equation}
\end{theorem}

\begin{proof}
  Any symplectic form admits a (local) potential $\psi = - (L_{n})$ with respect
  to the differential $e_1$ in the variational sequence
  \cite{Many,Dorf,hand}. This means in coordinates
  \begin{equation}
    B_{in}= - \left(\pd{L_i}{\varphi^n_\sigma}\partial _{\sigma } -
    (-1)^{\sigma }\partial _{\sigma }\pd{L_n}{\varphi^i_\sigma}
     \right),  \label{eq:43}
  \end{equation}
  Then, it is easy to prove that
  \begin{equation}
    \fd{(L_n\varphi^n_t)}{\varphi^i}=(-1)^{\sigma }\partial _{\sigma }\left( %
      \pd{L_n}{\varphi^i_\sigma}\varphi _{t}^{n}\right) -\pd{L_i}{\varphi^n_\sigma}%
    \partial _{\sigma }\varphi _{t}^{n} = B_{in}\varphi _{t}^{n},  \label{eq:44}
  \end{equation}
  so that $L_{n}\varphi_{t^{k}}^{n}-h_{k+2}$ is the Lagrangian of~\eqref{eq:41}.

  Let us recall that $\ell_\psi = - \pd{L_i}{\varphi^n_\sigma}\partial _{\sigma
  }$. Then, eq.~\eqref{eq:43} can be rewritten as
  \begin{equation}\label{eq:8}
    B=\ell_\psi - \ell^*_\psi
  \end{equation}
  in potential variables $\varphi^i_x$. Changing coordinates to $u^i$ yields
  the following changes of variables:
  \begin{displaymath}
    \ell_\psi(\varphi^i_{x},\varphi^i_{xx},\ldots)=
    \ell_\psi(u^i,u^i_{x},\ldots)\circ \partial_x,
    \quad
    \ell_\psi^*(\varphi^i_{x},\varphi^i_{xx},\ldots)=
    - \partial_x \circ \ell_\psi^*(u^i,u^i_{x},\ldots).
  \end{displaymath}
  After multiplying~\eqref{eq:9} to the left and right by $K$ we have, in
  coordinates $(u^i)$,
  \begin{equation}\label{eq:11}
    A_2^{ij}= - K^{ih}(\ell_\psi\circ\partial_x +
      \partial_x\circ\ell_\psi^*)_{hk}K^{kj}
  \end{equation}
  We would like to prove that $A_2=L_\tau A_1$. A natural candidate for $\tau$
  is $\tau=K\psi$: we have $\ell_{A_1,\psi}(\tau)=0$ because $A_1$ has constant
  coefficients, and $\ell_\tau=K\ell_\psi$, so that expanding the definition
  of the operator $L_\tau A_1$ by computing it on a covector field $\xi$ we
  have:
  \begin{align*}
    L_\tau A_1(\xi)&= - K\ell_\psi(A_1(\xi)) - A_1((K\ell_\psi)^*(\xi))
    \\
    &= - K^{ih}(\ell_\psi\circ\partial_x +
      \partial_x\circ\ell_\psi^*)_{hk}K^{kj}
    \\
    &=A_2(\xi).\qedhere
  \end{align*}
\end{proof}

The above Theorem has the following straightforward consequence.

\begin{corollary}
  Finding the Lagrangian representation~\eqref{eq:42} of the integrable
  hierarchy~\eqref{eq:381} is equivalent to finding the generalized vector
  field $\tau $ such that $A_2 = L_\tau A_1$.
\end{corollary}

\begin{remark}\label{sec:lagr-repr-triv}
  The vector function $\psi = - (L_n)$ depends on derivatives of $(\varphi^i)$
  which are of the same order of those appearing in $A_2$ (in potential
  coordinates). This is due to the fact that the homotopy operator of the
  variational sequence does not change the order of derivatives. Theorems and
  conjectures about minimising the order in inverse problems of the calculus of
  variations are discussed in \cite{hand}.
\end{remark}

\section{Lagrangian representation for homogeneous
 bi-Hamiltonian pairs}

It is now clear that the explicit expression of $\tau$ also yields the
Lagrangian representation of the given bi-Hamiltonian system of PDEs.

Following \cite{Dorf,DMS05,DZ0}, $\tau$ can be computed through $\psi$, and
$\psi$ can be computed using the homotopy operator in the variational
sequence. The vector field $\tau$ plays the role of the \emph{infinitesimal
  deformation} in the perturbative approach of \cite{DZ0}. The approach works
in the simplest examples like the KdV equation \cite{DMS05}; however, if the
range of examples is extended to even just slightly more complicated examples
the homotopy operator approach is no longer effective. Indeed the integrand in
the homotopy operator is singular for $t=0$ in many examples (see
Section~\ref{sec:lagr-repr-wdvv}).  Such singularities are not removable by a
shift: the shift operator does not bring total divergencies into total
divergencies, hence shifting does not preserve the homotopy operator in the
variational sequence (see \cite[p.\ 65]{Dorf}).

In this Section we propose an alternative approach to the computation of $\tau$
when $A_2$ is a homogeneous operator of Dubrovin--Novikov type \cite{DN2}. The
approach will be developed in details when $A_2$ is of the third order. Indeed,
there is a huge family of examples of systems of PDEs which are bi-Hamiltonian
with respect to a pair of homogeneous operators $A_1$ and $A_2$, with $A_1$ of
the first order and $A_2$ of the third order, see
Section~\ref{sec:lagr-repr-wdvv}. However, our methods easily extend to
homogeneous operators of arbitrary order.

Under the above hypotheses $A_2$ has the form
\begin{multline}
A_{2}^{ij}=g_{2}^{ij}(\mathbf{u})\partial _{x}^{3}+b_{2k}^{ij}(\mathbf{u}
)u_{x}^{k}\partial _{x}^{2}+[c_{2k}^{ij}(\mathbf{u})u_{xx}^{k}+c_{2km}^{ij}(
\mathbf{u})u_{x}^{k}u_{x}^{m}]\partial _{x}  \label{eq:17} \\
+d_{2k}^{ij}(\mathbf{u})u_{xxx}^{k}+d_{2km}^{ij}(\mathbf{u}
)u_{xx}^{k}u_{x}^{m}+d_{2kmn}^{ij}(\mathbf{u})u_{x}^{k}u_{x}^{m}u_{x}^{n}.
\end{multline}
We assume that $A_2$ is non-degenerate, \emph{i.e.} $\det{g}^{ij}\neq 0$.
Operators of the above type can always be brought by a point transformation of
the type $\tilde{u}=\tilde{u}(\mathbf{u})$ into the canonical form
\cite{BP,Doyle,GP97,GP91}
\begin{equation}
A_{2}^{ij}=\partial _{x}^{{}}( g^{ij}\partial
_{x}^{{}}+c_{k}^{ij}a_{x}^{k}) \partial _{x}^{{}}.  \label{casimir}
\end{equation}
Such operators have been recently studied and classified \cite{fpv,fpv2}.

Now we devote ourselves to finding the Lagrangian representation.  First of all
we observe that $L_n=L_n(u^k,u^k_x,u^k_{xx},u^k_{xxx})$.  Then, since homotopy
operator in the variational sequence preserves homogeneity, the functions $L_n$
must be a homogeneous polynomials with respect to the derivative
coordinates. We have another interesting property of each of the functions
$L_n$.

\begin{lemma}
The components $L_n$ of the Lagrangian representation~\eqref{eq:42} are
conservation law densities of the initial system~\eqref{eq:381}.
\end{lemma}
\begin{proof}
  We can start the recurrence relation~\eqref{eq:381} from $\mathbf{H}%
  _{0}^{(k)}=\int u^{k}dx$: we have $\mathbf{H}_{1}^{(k)}=K^{km}\int
  L_{m}dx$. Indeed, the expression $A_2^{ij}\tilde{\psi}_j$ can be expanded
  using formula~\eqref{eq:11}; we have
  \begin{align*}
    A_2^{ij}\tilde{\psi}_j
      =& K^{ih} \left(
        \pd{L_h}{u^k_\sigma}\partial _{\sigma}\partial_x
          \left(K^{kj}\tilde{\psi}_j\right)
+
    \partial_x(-1)^{\sigma }\partial _{\sigma
      }\left(\pd{L_k}{u^h_\sigma}K^{kj}\tilde{\psi}_j\right)
\right)
    \\
      =& K^{ih}\partial_x\left(\mathcal{E}(L_k)_h\right)K^{kj}\tilde{\psi}_j
         + \text{higher order terms,},
  \end{align*}
  where `higher order terms' are terms linear in $\partial_\sigma\tilde{\psi}_j$
  and
  \begin{displaymath}
    \mathcal{E}(L_k)_h = (-1)^{\sigma }\partial _{\sigma
      }\left(\pd{L_k}{u^h_\sigma}\right)
  \end{displaymath}
  is the Euler-Lagrange expression. Using the above expression in the
  recurrence relation yields the proof.
\end{proof}

The above Lemma does not help us in finding a representative of $L_n$ of
minimal order (up to total divergencies). The minimisation of order in inverse
problems of the Calculus of Variations is a mathematical area which is richer
in conjectures than in results, see \cite{hand} for a review on these
problems. However, in concrete examples homogeneity considerations on the
recurrence relation tells us that $L_n$ can be a homogeneous polynomial of the
derivative coordinates of degree not less than $2$.  Every homogeneous
polynomial of degree $2$ has the form
$L_n=L^1_{nmp}u^m_xu^p_x+L^2_{nm}u^n_{xx}$, but we prefer to rewrite $L_n$ as
follows:
\begin{equation}
L_{n}=\left( \frac{1}{2}G_{nm}u_{x}^{m}+R_{nm}u_{x}^{m}\right) _{x}-\frac{1}{%
2}L_{nsm}u_{x}^{s}u_{x}^{m}.  \label{eq:46}
\end{equation}%
where $G_{nm}=G_{mn}$, $R_{mn}=-R_{nm}$, $L_{nsm}=L_{nms}$. Indeed, examples
show that we can actually find conservation law densities of the above
form. Note that the $n$ conservation law densities
$-\frac{1}{2}L_{nsm}u_{x}^{s}u_{x}^{m}$ are not independent in general (see
\cite{fs}). We need the corresponding form of $A_2$; it can be obtained by a
straightforward computation.
\begin{lemma}
  If the Lagrangian representation $(L_n)$ has the form \eqref{eq:46} then the
  operator $A_2$ takes the form
  \begin{multline} \label{eq:47} A_2^{ij}=K^{ip}G_{pn}K^{nj}\partial_{x}^{3} +
    K^{ip}(F_{pmn}+G_{pn,m}-L_{npm}-L_{pnm})K^{nj}u_{x}^{m}\partial_{x}^{2} \\
    +K^{ip}\left[F_{pmn}u_{xx}^{m}+\left( F_{pmn,s}-\frac{1}{2}L_{psm,n}-
        \frac{1%
        }{2}L_{nsm,p}\right) u_{x}^{s}u_{x}^{m}\right] K^{nj}\partial _{x} \\
    +K^{ip}\left[ L_{npm}u_{xx}^{m}+L_{npm,s}u_{x}^{s}u_{x}^{m}- \frac{1}{2}
      L_{nsm,p}u_{x}^{s}u_{x}^{m}\right] _{x}K^{nj},
\end{multline}
where
\begin{equation}
  F_{pmn}=\frac{1}{2}G_{pm,n}+\frac{1}{2}G_{np,m}- \frac{1}{2}
G_{nm,p}+2L_{npm}+R_{pm,n}+R_{mn,p}+R_{np,m}.  \label{eq:48}
\end{equation}
\end{lemma}

Now we can prove the main result of this section.

\begin{theorem}
The Lagrangian representation~\eqref{eq:42} can be found by a Lagrangian of
the form given in~\eqref{eq:46} if and only if the tensor
\begin{equation}
  \label{eq:21}
  T_{pmn} =  F_{pmn} -
    \frac{1}{2}\left(G_{pm,n} + G_{np,m} - G_{nm,p}+4L_{npm}\right)
\end{equation}
is a closed $3$-form. In this case $G_{kp}$ and $L_{smn}$ can be uniquely
determined while $R_{kp}$ is determined up to the differential of a $1$-form.
\end{theorem}
\begin{proof}
If we assume that $L_n=L_n(u^k,u^k_x,u^k_{xx})$ then the
expression~\eqref{eq:11} reads as
\begin{multline}
A_{2}^{ij}=K^{ip}\left( \frac{\partial L_{p}}{\partial u_{xx}^{n}}+\frac{
\partial L_{n}}{\partial u_{xx}^{p}}\right) K^{nj}\partial _{x}^{3}
\\
+K^{ip}
\left[ \frac{\partial L_{p}}{\partial u_{x}^{n}}-\frac{\partial L_{n}}{
\partial u_{x}^{p}}+3\left( \frac{\partial L_{n}}{\partial u_{xx}^{p}}
\right) _{x}\right] K^{nj}\partial _{x}^{2}  \label{eq:45} \\
+K^{ip}\left[ \frac{\partial L_{p}}{\partial u^{n}}+\frac{\partial L_{n}}{%
\partial u^{p}}-2\left( \frac{\partial L_{n}}{\partial u_{x}^{p}}\right)
_{x}+3\left( \frac{\partial L_{n}}{\partial u_{xx}^{p}}\right) _{xx}\right]
K^{nj}\partial _{x} \\
+K^{ip}\left[ \frac{\partial L_{n}}{\partial u^{p}}-\left( \frac{\partial
L_{n}}{\partial u_{x}^{p}}\right) _{x}+\left( \frac{\partial L_{n}}{\partial
u_{xx}^{p}}\right) _{xx}\right] _{x}K^{nj},
\end{multline}
Then we can plug in the above formula the expression of $L_n$ given
in~\eqref{eq:46}. Keeping in mind the general expression of a homogeneous
third-order operator~\eqref{eq:17} we have:
\begin{enumerate}
\item $G_{kn} = M_{ks}g^{sp}_2M_{pn}$ by comparing the leading coefficient, so
  that $G_{hk}$ is uniquely determined;
\item $L_{smn} = M_{sk}d^{kp}_{2m} M_{pn}$ by comparing the coefficient of
  $u^m_{xxx}$, so that $L_{smn}$ is uniquely determined;
\item $F_{pmn} = M_{ip}c^{ij}_{2m}M_{jn}$ by comparing the coefficient of
  $u^m_{xx}$ in the coefficient of $\partial_x$, so that $F_{pmn}$ is uniquely
  determined.
\end{enumerate}
The equation~\eqref{eq:48} can be rewritten as
\begin{equation}\label{eq:12}
  R_{pm,n}+R_{mn,p}+R_{np,m} =
F_{pmn} - \frac{1}{2}\left(G_{pm,n} + G_{np,m} - G_{nm,p} + 4L_{npm}\right).
\end{equation}
If we interpret $R_{pm}$ as a $2$-form, then the above equation is integrable
as the right-hand side is a closed $3$-form. As our considerations are local,
we can solve the above equation (non-uniquely) by Poincar\'e's Lemma.
\end{proof}

\begin{corollary}
The matrix $G_{ij}$ in~\eqref{eq:46} is the leading term of the symplectic
operator $B$.
\end{corollary}

\begin{remark}
  In all our concrete examples (WDVV systems) the integrability requirement on
  the tensor $T$~\eqref{eq:21} is passed. We \emph{conjecture} that this
  condition is a consequence of the fact that one can always obtain a minimal
  order potential $\psi$ for any symplectic operator $B$. This conjecture goes
  along the lines of existing conjecture on minimal order Lagrangians for
  locally variational PDEs, see~\cite{hand} and references therein.
\end{remark}

\section{Lagrangian representation of WDVV equations}
\label{sec:lagr-repr-wdvv}

The main examples of bi-Hamiltonian systems of the above type that we have in
mind are the Witten--Dijkgraaf--Verlinde--Verlinde (WDVV) equations. The WDVV
system arises in $2D$ topological field theory \cite{Dub1,Dub2b} as the
associativity condition of an algebra in an $N$-dimensional space.  This system
is an overdetermined system of third-order PDEs in one unknown function
$f$. The form of the equations also depend on the choice of a nondegenerate
scalar product $\eta_{\alpha\beta}$ on $\mathbb{R}^N$.

The WDVV equations on an $N$-dimensional space can also be presented in the
form of $N-2$ hydrodynamic-type systems in $N(N-1)/2$ components. Such systems
are mutually commuting and non-diagonalizable (see\cite{OM98} and references
therein).  In the first non-trivial case $N=3$ it has been showed that for
three distinct choices of $\eta_{\alpha\beta}$ such systems of PDEs are
bi-Hamiltonian by a pair of local operators of the type \eqref{eq:192} and
\eqref{eq:17} (see \cite{FGMN,KN1,KN2}). The same has been proved more recently
for two systems in the case $N=6$ \cite{pv}.

In this section we will find the Lagrangian representation of the WDVV system
from \cite{FGMN}. As a by-product we will find a structure formula for the
third-order operator $A_2$ that has an independent interest.

Given a function $F=F(t^1,\ldots,t^N)$ we assume that
\begin{equation}
\eta _{\alpha \beta }=\frac{\partial ^{3}F}{\partial t^{1}\partial t^{\alpha
}\partial t^{\beta }}  \label{eq:1}
\end{equation}
is a constant nondegenerate symmetric matrix ($\eta^{\alpha \beta }$ will
denote its inverse matrix); the WDVV equations are equivalent to the
requirement that the functions
\begin{equation}
c_{\beta \gamma }^{\alpha }=\eta ^{\alpha \mu }
\frac{\partial ^{3}F}{\partial t^{\mu }\partial t^{\beta }\partial t^{\gamma }}
\end{equation}
are the structure constants of an associative algebra. Then the
associativity condition reads as
\begin{equation}
\eta^{\mu \lambda }
\frac{\partial ^{3}F}{\partial t^{\lambda }\partial t^{\alpha }\partial t^{\beta }}
\frac{\partial ^{3}F}{\partial t^{\nu}\partial t^{\mu }\partial t^{\gamma }}
=
\eta^{\mu \lambda }
\frac{\partial^{3}F}
{\partial t^{\nu }\partial t^{\alpha }\partial t^{\mu }}
\frac{\partial^{3}F}{\partial t^{\lambda }\partial t^{\beta }\partial t^{\gamma }}
\end{equation}

The integrability of the above equations was proved in \cite{Dub2b} by giving a
Lax pair for all values of $N$ and $\eta_{\alpha \beta }$. The Hamiltonian
geometry of WDVV equations also attracted the interest of a number of
researchers \cite{FGMN,KKVV,OM98}. We will focus on \cite{FGMN}, where the case
$N=3$ was considered with $\eta$ antidiagonal identity, \emph{i.e.} $\eta
_{\alpha \beta }=\delta _{\alpha +\beta ,4}$. In this case $F$ is of the form
$F=\frac{1}{2}(t^{1})^{2}t^{3}+\frac{1}{2} t^{1}(t^{2})^{2}+f(t^{2},t^{3})$,
and the WDVV system consists of the single equation (after setting $ x=t^{2}$,
$t=t^{3}$)
\begin{equation}
f_{ttt}=f_{xxt}^{2}-f_{xxx}f_{xtt}.  \label{eq:4}
\end{equation}
The above equation is a third order Monge-Amp\`ere equation (see \cite{mm} for
recent results).

Let us introduce the new variables $a^{1}=a=f_{xxx}$, $a^{2}=b=f_{xxt}$,
$a^{3}=c=f_{xtt}$. Then the compatibility conditions for the WDVV equation
can be written as an hydrodynamic type system of PDEs
$a_{t}^{i}=v_{j}^{i}(\mathbf{a})a_{x}^{j}$ where
\begin{equation}
a_{t}=b_{x}, \quad b_{t}=c_{x}, \quad c_{t}=(b^{2}-ac)_{x}.  \label{eq:5}
\end{equation}

It was proved in \cite{FGMN} that the above system can be written as a
Hamiltonian system in two ways:
\begin{equation}
a_{t}^{i}=A_{1}^{ij}\fd{H_2}{a^j}=A_{2}^{ij}\fd{H_{1}}{a^j}  \label{eq:6}
\end{equation}%
with respect to two compatible local Hamiltonian operators $A_{1}$ and
$A_{2}$, with expressions
\begin{gather}
A_{1}=%
\begin{pmatrix}
-\frac{3}{2}\partial _{x}^{{}} & \frac{1}{2}\partial _{x}^{{}}a & \partial
_{x}^{{}}b \\
\frac{1}{2}a\partial _{x}^{{}} & \frac{1}{2}(\partial _{x}^{{}}b+b\partial
_{x}^{{}}) & \frac{3}{2}c\partial _{x}^{{}}+c_{x} \\
b\partial _{x}^{{}} & \frac{3}{2}\partial _{x}^{{}}c-c_{x} &
(b^{2}-ac)\partial _{x}^{{}}+\partial _{x}^{{}}(b^{2}-ac)%
\end{pmatrix}
\\
A_{2}=%
\begin{pmatrix}
0 & 0 & \partial _{x}^{3} \\
0 & \partial _{x}^{3} & -\partial _{x}^{2}a\partial _{x} \\
\partial _{x}^{3} & -\partial _{x}a\partial _{x}^{2} & \partial
_{x}^{2}b\partial _{x}+\partial _{x}b\partial _{x}^{2}+\partial
_{x}a\partial _{x}a\partial _{x}%
\end{pmatrix}%
\end{gather}%
and Hamiltonian densities $h_{2}=c$, $h_{1}=-\frac{1}{2} a(\partial
_{x}^{-1}b)^{2}-(\partial _{x}^{-1}b)(\partial _{x}^{-1}c)$, respectively.  The
two Hamiltonian operators $A_{1}$ and $A_{2}$ are homogeneous (see
\cite{DN,DN2} and the Introduction for more details).

The observation that led to finding
$A_{1}$ was that in the Lax pair of the system \eqref{eq:5}
\begin{equation}
  \label{eq:15}
  \psi_x = \lambda
  \begin{pmatrix}
    0 & 1 & 0\\ b & a & 1 \\ c & b & 0
  \end{pmatrix}
  \psi,\qquad
  \psi_t = \lambda
  \begin{pmatrix}
    0 & 0 & 1\\ c & b & 0 \\ b^2-ac & c & 0
  \end{pmatrix}
  \psi
\end{equation}
the eigenvalues $u^{k}(\mathbf{a})$ of the matrix that gives the $x$-evolution
are conservation law densities. If the system is rewritten using the above
eigenvalues as new dependent variables $u^{k}$, \emph{i.e.}, using the point
transformation
\begin{equation}
  a=u^1+u^2+u^3,
\quad
  b=-\frac{1}{2}(u^1u^2+u^2u^3+u^3u^1),
  \quad
  c=u^1u^2u^3
\end{equation}
the operator $A_{1}$ becomes evident and is of the type
$A_{1}^{ij}=K^{ij}\partial _{x}^{{}}$, where
\begin{equation*}
K=\frac{1}{2}\left(
\begin{array}{ccc}
1 & -1 & -1 \\
-1 & 1 & -1 \\
-1 & -1 & 1%
\end{array}%
\right),
\end{equation*}
and the Hamiltonian is $\mathcal{H}_1=u^1u^2u^3$.  In these new coordinates our
system~\eqref{eq:5} takes the form
\begin{equation}
  \label{eq:3}
    u^i_t=\frac{1}{2}(u^ju^k-u^iu^j-u^iu^k)_x
\end{equation}
where $i$, $j$, $k$ are three distinct indices.

The operator $A_{2}$ for the system \eqref{eq:5} was found in a
completely different way. More precisely, a Lagrangian for the $x$-derivative
of the WDVV equation \eqref{eq:4} was found, and a symplectic representation of
this equation was achieved in \cite{FGMN}. Then $A_{2}$ was found by
inverting a corresponding symplectic form and multiplying it by
$A_{1}$. It is necessary to emphasize that the coordinates $a$, $b$, $c$ (see
\eqref{eq:5}) are Casimirs for $A_{2}$, \emph{i.e.} $A_{2}$ is in the canonical
form~\eqref{casimir} in these coordinates, and the inverse matrix of the
leading term $g^{ik}$ is the Monge metric \cite{fpv}
\begin{equation}
g_{ij}=
\begin{pmatrix}
-2b & a & 1 \\
a & 1 & 0 \\
1 & 0 & 0%
\end{pmatrix}
\label{eq:7}
\end{equation}

At this point we would like to stress that \emph{flat coordinates $u^i$ of the
  first operator $A_1$ make the expression of the second operator $A_2$ much
  more complicated with respect to the initial coordinates $a^i$, since the
  latter are Casimirs of $A_2$}. On the other hand, our structure
formula~\eqref{eq:47} connects the operator $A_2$ with conservation law
densities of our system, with the leading coefficient of the symplectic
operator and with a newly introduced skew-symmetric tensor $R_{mn}$, and all
these quantities are computable in terms of flat coordinates of the first
operator, thus leading to the Lagrangian representation~\eqref{eq:13}.  This
means that we have to change coordinates to the operator $A_2$. To this aim we
use the following formula:
\begin{equation}
  \label{eq:16}
  A_2^{ij}(\mathbf{u})=\frac{\partial u^{i}}{\partial a^{n}}
   A^{nm}_2(\mathbf{a})\frac{\partial u^{j}}{\partial a^{m}},
\end{equation}
The leading coefficient of $A_2$ is the following contravariant
metric $g^{ij}(\mathbf{u})$
\begin{align}
  g^{ii} & =
    \frac{3{u^i}^2+{u^j}^2+{u^k}^2
    -3u^iu^j-3u^iu^k+u^ju^k}
    {(u^i-u^j)^2(u^i-u^k)^2}\notag
  \\
    & = \frac{(u^i-u^j)^2 + (u^i-u^k)^2 + (u^j-u^k)^2
      +4(u^i-u^j)(u^i-u^k)}{(u^i-u^j)^2(u^i-u^k)^2},
  \\
    g^{ij} & = \frac{-1}{(u^i-u^j)^2},
\end{align}
where $i$, $j$, $k$ are three pairwise distinct indexes.
The corresponding covariant metric is
\begin{align}
  g_{ii} & =
  \frac{-3(u^j-u^k)^2}{4},
  \\
  g_{ij} & =
  \frac{1}{4}(-2{u^i}^2-2{u^j}^2-3{u^k}^2
  +u^iu^j+3u^iu^k+3u^ju^k)\notag
  \\
  & = -\frac{1}{4}((u^i-u^j)^2 + (u^i-u^k)^2 + (u^j-u^k)^2
      + (u^i-u^j)(u^i-u^k)),
\end{align}
where $i$, $j$, $k$ are three pairwise distinct indexes.

The leading term $G_{ij}(\mathbf{u})$ of the symplectic operator is
\begin{align}
G_{ii}&=\frac{
(u^{i})^{2}+(u^{j})^{2}+(u^{k})^{2}-u^{i}u^{j}-u^{i}u^{k}-u^{j}u^{k}}{%
(u^{i}-u^{j})^{2}(u^{i}-u^{k})^{2}}\notag   \\
&=\frac{(u^i-u^j)^2 + (u^i-u^k)^2 + (u^j-u^k)^2}
{2(u^{i}-u^{j})^{2}(u^{i}-u^{k})^{2}},  \label{eq:51}
\\
G_{ij} & =\frac{(u^{i})^{2}+(u^{j})^{2}-(u^{k})^{2}+u^{i}u^{k}+u^{j}u^{k}-3u^{i}u^{j}}
{(u^{i}-u^{j})^{2}(u^{i}-u^{k})(u^{j}-u^{k})} \notag
\\
& = \frac{(u^i-u^j)^2 + (u^i-u^k)^2 + (u^j-u^k)^2 - 4(u^k-u^i)(u^k-u^j)}
{2(u^{i}-u^{j})^{2}(u^{i}-u^{k})(u^{j}-u^{k})},
\end{align}%
where $i$, $j$, $k$ are a triplet of distinct indices. Note that
\begin{equation}
\det \mathbf{G}=-\frac{16}{%
(u^{1}-u^{2})^{2}(u^{1}-u^{3})^{2}(u^{2}-u^{3})^{2}},  \label{eq:52}
\end{equation}%
and from $G_{11}>0$ (outside obvious singularities) it follows that the
signature of the metric is $(2,1)$.  It is also remarkable that the metric
$G_{ij}$ has constant sectional curvature $-1/16$. The corresponding
contravariant metric is
\begin{align*}
  G^{ii}&=-\frac{1}{4}(u^i-u^j)(u^i-u^k),
  \\
  G^{ij}&=-\frac{1}{4}(u^i-u^j)^2.
\end{align*}

The expressions of $L_{ijk}=L_{ikj}$ are:
\begin{enumerate}
\item when $j\neq 1$ and $k\neq 1$ or when $j=k=1$
\begin{equation}
  \label{eq:18}
  L_{1jk}=\frac{((u^1-u^2)+(u^1-u^3))(u^{a}-u^{b})(u^{c}-u^{d})}
        {2(u^1-u^2)^3(u^1-u^3)^3}
\end{equation}
where $(a,j,b)$, $(c,d,k)$ are triplets of distinct indices with $a<b$, $c<d$;
\item when $(1,j,k)$ are a triplet of distinct indices.
\begin{equation}
  \label{eq:190}
  L_{11k}= - \frac{(u^{1}-u^{k})^2+(u^{1}-u^{j})^2}
        {2(u^1-u^{k})^2(u^1-u^{j})^3}
\end{equation}
\item when $i\neq 1$ the expressions of $L_{ijk}$ are obtained by a cyclic
  permutation of the above expressions.
\end{enumerate}

\begin{remark}
  The conserved quantities $\mathcal{L}%
  _{n}=-\frac{1}{2}L_{nmp}u_{x}^{m}u_{x}^{p}$ coincide with the conserved
  quantities $\mathcal{I}_{1}$, $\mathcal{I}_{2}$, $\mathcal{I}_{3}$ defined in
  the paper \cite{FGMN}: $\mathcal{L}_{n}=\mathcal{I}_{n}$.
\end{remark}

\begin{remark}
  There is a further conservation law density in the form $u^iL_i$; however,
  this density is trivial as it is in the form of a total divergence.
\end{remark}

The integrability condition is fulfilled in our case as $T_{pmn}$ is a
completely skew-symmetric tensor, and by dimensional reasons. The single
equation for $R_{mn}$ is
\begin{equation*}
\left( \pd{R_{31}}{u^2}+\pd{R_{23}}{u^1}+\pd{R_{12}}{%
u^3}\right) (u^{1}-u^{2})(u^{1}-u^{3})(u^{2}-u^{3})+1=0.
\end{equation*}%
in the unknown functions $R_{23}$, $R_{13}=-R_{31}$, $R_{12}$ of the
coordinates $u^{1}$, $u^{2}$, $u^{3}$. This equation can be regarded as an
equation of the form $dR=-1/((u^{1}-u^{2})(u^{1}-u^{3})(u^{2}-u^{3}))du^{1}%
\wedge du^{2}\wedge du^{3}$, where $R=R_{12}du^{1}\wedge
du^{2}+R_{13}du^{1}\wedge du^{3}+R_{23}du^{2}\wedge du^{3}$. One
distinguished solution can be obtained by letting $\pd{R_{31}}{u^2}=%
\pd{R_{23}}{u^1}=\pd{R_{12}}{u^3}$. In this case we have
\begin{equation}
R_{ij}^{0}=-\frac{1}{3}\left( \frac{1}{(u^{i}-u^{j})(u^{i}-u^{k})}-\frac{1}{%
(u^{j}-u^{i})(u^{j}-u^{k})}\right)  \label{eq:53}
\end{equation}%
where $i$, $j$, $k$ are distinct indices.

It is possible to cross-check the above results by a direct calculation of the
quantity $L_\tau A_1$ aimed at verifying that $L_\tau A_1=A_2$. This has been
done using the correspondence between multivectors and superfunctions
in~\cite{getz,IgoninVerbovetskyVitolo:VMBGJS} and the corresponding Schouten
bracket formula from~\cite{KerstenKrasilshchikVerbovetsky:HOpC}. The result was
confirmed after 18 hours of computation with CDE, using 7.7GB of RAM, on the
workstation \texttt{sophus2} of the Department of Mathematics and Physics of
the Universit\`a del Salento.

\section{Conclusions}

A bi-Hamiltonian structure for WDVV equation with different $\eta _{\alpha
  \beta }$ in the case $N=3$ also has been considered. For instance, in
\cite{KN1} a different identification of the variables $t^{2}$ and $t^{3}$ as
$t$ and $x$ leads to different WDVV equations and different bi-Hamiltonian
formulations through local Dubrovin--Novikov operators. Moreover, in \cite{KN2}
another choice of constants in $\eta _{\alpha \beta }$ was investigated. Its
third order operator of \cite{KN2} lies in a different class with respect to
the third order operator of \cite{FGMN}, with respect to the classification in
\cite{fpv}. Since in the above cases different choices of $\eta$ can be
connected by transformations of independent variables to the choice of the WDVV
equation that we considered~\eqref{eq:4}, we expect that the Lagrangian
representation can be effectively achieved by the techniques exposed in this
paper.

We stress that Hamiltonian operators for the $N=3$ WDVV equation~\eqref{eq:4}
in the original unknown $f$ have been found in \cite{KKVV}. However, they are
nonlocal and depend explicitly on independent variables and it is not clear
(and a nontrivial problem) how to relate them to the operator in \cite{FGMN}.

When $N=4$ in~\eqref{eq:1} and choosing $\eta_{\alpha\beta}$ to be the
antidiagonal identity we obtain a system of equations in third-order
derivatives of $f$ that can be rewritten as two commuting hydrodynamic-type
systems in $6$ unknown functions. These systems have been recently proved to be
bi-Hamiltonian \cite{pv}, and a Lagrangian representation have been computed
with the above methods. However, its expression is more complicated and less
`regular' than the above expressions for $N=3$. For this reason we did not
decide to write down the expressions, even if we will make them available upon
request. Here the dimension of the problem does not automatically imply the
integrability condition $dT=0$, and the fact that this condition holds is
remarkable.

We stress that our methods might be carried out for an arbitrary homogeneous
operator $A_2$ with no conceptual differences with the third-order case.

\section*{Acknowledgements}

We would like to thank B.A. Dubrovin, E.V. Ferapontov, P. Lorenzoni and
G.V. Potemin for many interesting discussions. Thanks are due to A.C. Norman
for his support with Reduce. We also thank A. Falconieri for his technical
support with the workstation \texttt{sophus2} of the Dipartimento di Matematica
e Fisica of the Universit\`a del Salento. Finally, we acknowledge the anonymous
Referee for his/her comments that helped us to clarify our results.

MVP and RFV were partially supported by GNFM of the Istituto Nazio\-nale di
Alta Matematica, the PRIN ``Geometric and analytic theory of Hamiltonian
systems in finite and infinite dimensions'', Istituto Nazio\-nale di Fisica
Nucleare (INFN) -- IS CSN4 \emph{Mathematical Methods of Nonlinear Physics} and
the Dipartimento di Matematica e Fisica \textquotedblleft E. De
Giorgi\textquotedblright\ of the Universit\`{a} del Salento.

MVP's work was also partially supported by the grant of Presidium of RAS
``Fundamental Problems of Nonlinear Dynamics'' and by the RFBR grant RFBR grant
14-01-00012.


\begin{thebibliography}{99}
\bibitem{BP} \emph{A.V. Balandin, G.V. Potemin},   On non-degenerate
differential-geometric Poisson brackets of third order, Russian Mathematical
Surveys \textbf{56} No. 5 (2001) 976-977.

\bibitem{Many} \emph{A. V. Bocharov, V. N. Chetverikov, S. V. Duzhin, N. G.
Khor{'}kova, I. S. Krasil{'}shchik, A. V. Samokhin, Yu.\ N. Torkhov, A. M.
Verbovetsky and A. M. Vinogradov}: Symmetries and Conservation Laws for
Differential Equations of Mathematical Physics, I. S. Krasil{'}shchik and A.
M. Vinogradov eds., Translations of Math. Monographs \textbf{182}, Amer.\
Math.\ Soc. (1999).

\bibitem{DMS05} \emph{L. Degiovanni, F. Magri, V. Sciacca}, On deformation
of Poisson manifolds of hydrodynamic type, Comm. Math. Phys. 253 (2005),
1--24.

\bibitem{Doyle} \emph{P.W. Doyle},   Differential geometric Poisson
bivectors in one space variable, J. Math. Phys. \textbf{34} No. 4 (1993)
1314-1338.

\bibitem{Dorf} \emph{I. Dorfman}, Dirac structures and integrability of
nonlinear evolution equations, John Wiley \& Sons, England, 1993.

\bibitem{Dub1} \emph{B.A. Dubrovin,} Integrable systems in topological field
theory, Nucl. Phys. B, \textbf{379} (1992) 627--689.

\bibitem{Dub2b} \emph{B.A. Dubrovin}, Geometry of 2D topological field
theories, Lecture Notes in Math. 1620, Springer-Verlag (1996) 120--348.

\bibitem{DN} \emph{B.A. Dubrovin and S.P. Novikov}, Hamiltonian formalism of
one-dimensional systems of hydrodynamic type and the Bogolyubov-Whitham
averaging method, Soviet Math. Dokl. \textbf{27} No. 3 (1983) 665--669.

\bibitem{DN2} \emph{B.A. Dubrovin and S.P. Novikov}, Poisson brackets of
hydrodynamic type, Soviet Math. Dokl. \textbf{30} No. 3 (1984), 651--2654.

\bibitem{DZ0} \emph{B.A. Dubrovin and Y. Zhang}, Normal forms of integrable
    PDEs, Frobenius manifolds and Gromov-Witten invariants, math.DG/0108160.

\bibitem{FGMN} \emph{E.V. Ferapontov, C.A.P. Galvao, O. Mokhov, Y. Nutku}, %
  Bi-Hamiltonian structure of equations of associativity in 2-d
topological field theory, Comm. Math. Phys. \textbf{186 }(1997) 649-669.

\bibitem{fpv} \emph{E.V. Ferapontov, M.V. Pavlov, R.F. Vitolo},
Projective-geometric aspects of homogeneous third-order Hamiltonian
operators, J. Geom.\ Phys.\ 85 (2014), 16--28.

\bibitem{fpv2} \emph{E.V. Ferapontov, M. V. Pavlov, R.F. Vitolo},
Towards the classification of homogeneous third-order Hamiltonian
operators, to appear in Int.\ Math.\ Res.\ Notices 2016. ArXiv:
\url{http://arxiv.org/abs/1508.02752}.

\bibitem{fs} \emph{E.V. Ferapontov, R. Sharipov}, On first-order
conservation laws for systems of hydrodynamic type equations, Theoretical
and Mathematical Physics, Vol. 108, No. I, 1996, 937--952.

\bibitem{getz} \emph{E. Getzler}, A Darboux theorem for Hamiltonian
operators in the formal calculus of variations, Duke J. Math. \textbf{111}
(2002), 535-560.

\bibitem{IgoninVerbovetskyVitolo:VMBGJS}
\emph{S.~Igonin, A.~Verbovetsky, R.~Vitolo,}
Variational multivectors and brackets in the geometry of jet spaces.
In \emph{Symmetry in Nonlinear Mathematical Physics. {P}art 3}, pages
  1335--1342. Institute of Mathematics of NAS of Ukraine, Kiev, 2003.

\bibitem{KN1} \emph{J. Kalayci, Y. Nutku}, Bi-Hamiltonian structure of a
WDVV equation in 2d topological field theory, Phys.\ Lett.\ A \textbf{227}
(1997), 177--182.

\bibitem{KN2} \emph{J. Kalayci, Y. Nutku}, Alternative bi-Hamiltonian
structures for WDVV equations of associativity, J. Phys.\ A: Math.\ Gen.\
\textbf{31} (1998) 723-734.

\bibitem{KerstenKrasilshchikVerbovetsky:HOpC}
\emph{P.~Kersten, I.~Krasil{'}shchik, A.~Verbovetsky,}
Hamiltonian operators and $\ell^*$-coverings.
J. Geom.\ Phys., \textbf{50} (2004), 273--302.

\bibitem{KKVV} \emph{P.~Kersten, I.~Krasil{'}shchik,
    A.~Verbovetsky, R. Vitolo}, On integrable structures for a generalized
  Monge-Ampere equation , Theor.\ Math.\ Phys.\ \textbf{128}, no.\ 2 (2012),
  600--615.

\bibitem{Lor} \emph{P. Lorenzoni}, Deformations of bihamiltonian structures of
  hydrodynamic type, J. Geom.\ Phys.\ \textbf{44} no.\ 2-3 (2002), 331--375;
\url{https://arxiv.org/abs/nlin/0108015}

\bibitem{mm} \emph{G. Manno, G. Moreno}, Meta-Symplectic Geometry of 3rd
Order Monge--Amp\`ere Equations and their Characteristics, SIGMA \textbf{12}
(2016), 032 (35 pages); \url{http://arxiv.org/abs/1403.3521}

\bibitem{OM98} \emph{O.I. Mokhov},   Symplectic and Poisson
structures on loop spaces of smooth manifolds, and integrable systems,
Russian Math. Surveys \textbf{53} No. 3 (1998) 515-622.

\bibitem{Yavuz} \emph{Y. Nutku, M.V. Pavlov}, Multi Lagrangians for
  Integrable Systems, Journal of Mathematical Physics 43 (3), 1441 (2002);
  \url{http://dx.doi.org/10.1063/1.1427765}.

\bibitem{olver} \emph{P. Olver}, Applications of Lie Groups to Differential
  Equations, Springer GTM 107, 1993.

\bibitem{pv} \emph{M.V. Pavlov, R.F. Vitolo}, On the bi-Hamiltonian Geometry of
  WDVV Equations, Lett.\ Math.\ Phys.\ \textbf{105} no.\ 8 (2015), 1135--1163.
 \url{http://arxiv.org/abs/1409.7647}.

\bibitem{GP97} \emph{G.V. Potemin},   On third-order Poisson
brackets of differential geometry, Russ. Math. Surv. \textbf{52} (1997)
617-618.

\bibitem{GP91} \emph{G.V. Potemin},   Some aspects of differential
geometry and algebraic geometry in the theory of solitons. PhD Thesis,
Moscow, Moscow State University (1991) 99 pages.

\bibitem{Art} \emph{A. Sergeev}, A Simple Way of Making a Hamiltonian System
Into a Bi-Hamiltonian One, Acta Applicandae Mathematicae 83: 183--197, 2004.

\bibitem{hand} \emph{R. Vitolo}, Variational sequences, in `Handbook of
Global Analysis', ed. D. Krupka, D. Saunders, Elsevier 2007.

\bibitem{cde} \emph{R. Vitolo}, CDE: a Reduce package for computations
in the goemetry of differential equations, software, user guide and examples
freely available at \url{http://gdeq.org} (Version 1.0, October 2014).
See also the Reduce website: \url{http://reduce-algebra.sourceforge.net/}
\end{thebibliography}
\end{document}